\documentclass[11pt]{article}
\usepackage[left=1in,top=1in,right=1in,bottom=1in,head=.1in,nofoot]{geometry}

\usepackage{setspace,url,bm,amsmath} 

\usepackage{titlesec} 
\titlelabel{\thetitle.\quad} 
\titleformat*{\section}{\bf\large\center\uppercase} 

\usepackage{graphicx} 
\usepackage{bbm}
\usepackage{latexsym}
\usepackage{caption}
\usepackage[margin=20pt]{subcaption}
\usepackage{hyperref}

\usepackage[table]{xcolor}

\newcommand{\GG}[1]{}

\usepackage{amsthm}
\theoremstyle{definition}

\newtheorem{prop}{Proposition}
\newtheorem{lemma}{Lemma}
\newtheorem{example}{Example}

\newtheorem{corollary}{Corollary}
\newtheorem*{corollary*}{Corollary}

\usepackage{natbib} 
\bibpunct{(}{)}{;}{a}{}{,} 

\usepackage{etoolbox} 
\apptocmd{\sloppy}{\hbadness 10000\relax}{}{} 

\usepackage{color}
\usepackage{listings}

\begin{document}
\doublespacing
\title{\bf On Randomization-based and Regression-based Inferences for $2^K$ Factorial Designs}
\author{Jiannan Lu\footnote{Address for correspondence: Jiannan Lu, Microsoft Corporation, One Microsoft Way, Redmond, Washington 98052, USA.
Email: \texttt{jiannl@microsoft.com}}~\\Microsoft Corporation}
\date{}
\maketitle
\begin{abstract}
We extend the randomization-based causal inference framework in \cite{Dasgupta:2015} for general $2^K$ factorial designs, and demonstrate the equivalence between regression-based and randomization-based inferences. Consequently, we justify the use of regression-based methods in $2^K$ factorial designs from a finite-population perspective.
\end{abstract}
\textbf{Keywords:} Causal inference; potential outcome; unbalanced design; Huber-White estimator.

\section{Introduction}\label{sec:intro}

Factorial designs, originally introduced for agricultural experiments \citep{Fisher:1935, Yates:1937}, have gained more popularity in recent times because of their abilities to investigate multiple treatment factors simultaneously. As pointed out by \cite{Ding:2014}, although rooted in randomization theory \citep[e.g.,][]{Kempthrone:1952}, factorial designs have been dominantly analyzed by regression methods in practice. Unfortunately, however, regression-based inference might not be suitable under certain circumstances. For example, several researchers \citep[e.g.,][]{Miller:2006, Lu:2015} have pointed out that in many randomized experiments we cannot treat the experimental units as a random sample drawn from a hypothetical super-population, and should instead restrict the scope of inference to the finite-population of the experimental units themselves. Realizing the inherent deficiencies of regression-based inference, \cite{Dasgupta:2015} advocated conducting randomization-based inference for factorial designs by utilizing the concept of potential outcomes \citep{Neyman:1923, Rubin:1974}. The proposed framework for balanced $2^K$ factorial designs is flexible, interpretable and applicable to both finite-population and super-population settings. 

Given the advantages of randomization-based inference, it is necessary to generalize the framework in \cite{Dasgupta:2015} for more general, i.e., unbalanced, $2^K$ factorial designs. Moreover, it is of great importance to reconcile randomization-based and regression-based inferences, i.e., the point estimators of the factorial effects and their corresponding confidence regions. However, although the equivalence between randomization-based and regression-based inferences for randomized treatment-control studies (i.e., $2^1$ factorial designs) has been well established in the existing literature \citep{Schochet:2010, Samii:2012, Lin:2013}, similar discussions for $2^K$ factorial designs appear to be absent. In this paper, we fulfill the aforementioned two-fold task. 

The paper proceeds as follows. Section \ref{sec:2k} extends the randomization-based inference framework in \cite{Dasgupta:2015} to general $2^K$ factorial designs. Section \ref{sec:reg} demonstrates the equivalence between randomization-based and regression-based inferences for $2^K$ factorial designs. Section \ref{sec:extension} considers extensions, and Section \ref{sec:discuss} concludes and discusses possible future directions. 

\section{randomization-based inference for general $2^K$ factorial designs}\label{sec:2k}

\subsection{$2^K$ factorial designs}

Consider $K$ distinct factors, each with two levels -1 and 1. We construct the model matrix \citep{Wu:2009} $\bm H = (\bm h_0, \ldots, \bm h_{2^K-1})$ as follows: 
\begin{itemize}
\item let $\bm h_0 = \bm 1_{2^K};$ 
\item for $k=1,\ldots,K$, construct $\bm h_k$ by letting its first $2^{K-k}$ entries be -1, the next $2^{K-k}$ entries be 1, and repeating $2^{k-1}$ times;
\item for $k=K+1, \ldots, K+{K \choose 2},$ let $\bm h_k = \bm h_{k_1} \cdot \bm h_{k_2},$ where $k_1, k_2 \in \{1, \ldots, K \};$
\item[] $\ldots$
\item let $\bm h_{2^K-1}=\bm h_1\cdot \ldots \cdot \bm h_K.$
\end{itemize}
For $j=1, \dots, 2^K,$ let $\tilde{\bm h}_{j-1}$ denote the $j$th row of the model matrix $\bm H.$ A well-known fact is that the model matrix $\bm H$ is orthogonal, i.e.,
\begin{equation}\label{eq:orth}
\bm H \bm H^\prime = ( \tilde{\bm h}_j \tilde{\bm h}_{j^\prime}^\prime )_{0 \le j, j^\prime \le 2^k-1} = 2^K \bm I_{2^K},
\quad
\bm H^\prime \bm H = \sum_{j=1}^{2^K} \tilde{\bm h}_j^\prime \tilde{\bm h}_j = 2^K \bm I_{2^K}
\end{equation} 
The $j$th row of $\tilde{\bm H} = (\bm h_1, \ldots, \bm h_K)$ is the $j$th treatment combination $\bm z_j,$ and the columns of $\bm H$ define the factorial effects. To be specific, the first column $\bm h_0$ corresponds to the null effect, the next $K$ columns $\bm h_1, \ldots, \bm h_K$ correspond to the main effects of the $K$ factors, the next ${K \choose 2}$ columns $\bm h_{K+1}, \ldots, \bm h_{K+{K \choose 2}}$ correspond to the two-way interactions et. al., and eventually the last column $\bm h_{2^K-1}$ corresponds to the $K$-factor interaction. 

\begin{example}
For $2^2$ factorial designs, the model matrix is:
\begin{equation*}
\bm H =
\bordermatrix{& \bm h_0 & \bm h_1& \bm h_2 & \bm h_3\cr
             \tilde{\bm h}_0 & 1  & -1 &  -1  & 1 \cr
             \tilde{\bm h}_1 & 1 & -1 &  1  & -1 \cr
             \tilde{\bm h}_2 & 1  & 1 &  -1  & -1 \cr
             \tilde{\bm h}_3 & 1 & 1 & 1 & 1}.
\end{equation*}
The four treatment combinations are $\bm z_1=(-1, -1),$ $\bm z_2=(-1, 1),$ $\bm z_3=(1, -1)$ and $\bm z_4=(1, 1).$ We represent the main effects of factors 1 and 2 by $\bm h_1 = (-1,-1,1,1)^\prime$ and $\bm h_2=(-1,1,-1,1)^\prime$ respectively, and the two-way interaction by $\bm h_3=(1,-1,1,-1)^\prime.$ 
\end{example}

\subsection{Randomization-based Inference}

For consistency, we adopt the notations in \cite{Dasgupta:2015}. Let $N \ge 2^{K+1}$ be the number of experimental units. Under the Stable Unit Treatment Value Assumption \citep{Rubin:1980}, for unit $i,$ we denote its potential outcome under treatment combination $\bm z_j$ as $Y_i(\bm z_j),$ for $j=1, \ldots, 2^K.$ Let $\bm Y_i = \{ Y_i(\bm z_1), \ldots, Y_i(\bm z_{2^K}) \}^{\prime},$ and we define the factorial effect vector of unit $i$ as
\begin{equation}\label{eq:ie}
\bm \tau_i = \frac{1}{2^{(K-1)}} \bm H^\prime {\bm Y}_i.
\end{equation}
Having defined the potential outcomes and factorial effects on the individual-level, we shift focus to the population-level. For all $j,$ we let
$$
\bar Y(\bm z_j) = \frac{1}{N} \sum_{i=1}^N Y_i(\bm z_j)
$$
be the average potential outcome under treatment combination $\bm z_j,$ across all experimental units.
Let 
$
\bar{\bm Y} = \{ \bar Y(\bm z_1), \ldots, \bar Y(\bm z_{2^K}) \}^\prime,
$
and we define the population-level factorial effect vector as
\begin{equation}
\label{eq:fe}
\bm \tau = \frac{1}{N} \sum_{i=1}^N \bm \tau_i = \frac{1}{2^{(K-1)}} \bm H^\prime \bar{\bm Y}. 
\end{equation}

In this paper we consider general $2^K$ factorial designs, where we randomly assign $n_j \ge 2$ units to treatment $\bm z_j,$ for $j=1, \ldots, 2^K.$ Note that 
$
\sum_{j=1}^{2^K}  n_j = N.
$
For unit $i,$ we let 
$$
W_i(\bm z_j)
= 
\begin{cases}
1, & \text{if unit $i$ is assigned treatment $\bm z_j,$ } \\
0, & \text{otherwise.} \\
\end{cases}
$$
The observed outcome of unit $i$ is
$$
Y_i^\textrm{obs} = \sum_{j=1}^{2^K} W_i(\bm z_j) Y_i(\bm z_j).
$$
Let
$
\bm Y^\textrm{obs} = (Y_1^\textrm{obs}, \ldots, Y_N^\textrm{obs})^\prime
$
be the vector of all observed outcomes, and
$$
\bar Y^\textrm{obs}(\bm z_j) 
=
\frac{1}{n_j} \sum_{i:W_i(\bm z_j)=1} Y_i^{\textrm{obs}} 
=
\frac{1}{n_j} \sum_{i=1}^N W_i(\bm z_j) Y_i(\bm z_j).
$$
be the average observed outcome across all experimental units assigned to treatment combination $\bm z_j,$ and 
$
\bar{\bm Y}^\textrm{obs} = \{ \bar Y^\textrm{obs}(\bm z_1) , \ldots, \bar Y^\textrm{obs}(\bm z_{2^K}) \}^\prime.
$
\cite{Dasgupta:2015} defined the randomization-based estimator for $\bm \tau$ as
\begin{equation}
\label{eq:est}
\hat {\bm \tau}_{\textrm{RI}} =  \frac{1}{2^{(K-1)}} \bm H^\prime \bar{\bm Y}^\textrm{obs}, 
\end{equation}
whose randomness is solely from the treatment assignments. 

The following lemma plays an important role in deriving the sampling mean and covariance of the randomization-based estimator, and is also of independent interest. It is a slight modification of Lemma 4 in \cite{Dasgupta:2015}, and therefore we omit its proof.

\begin{lemma}\label{lemma:cov} 
Let the variance of potential outcomes for treatment $\bm z_j$ be
$$
S^2(\bm z_j) = \frac{1}{N-1} \sum_{i=1}^N \{ Y_i(\bm z_j) - \bar Y(\bm z_j) \}^2,
$$
and the covariance of potential outcomes for treatments $\bm z_j$ and $\bm z_{j^\prime}$ be 
$$
S(\bm z_j, \bm z_{j^\prime}) = \frac{1}{N-1} \sum_{i=1}^N \{ Y_i(\bm z_j) - \bar Y(\bm z_j) \} \{ Y_i(\bm z_{j^\prime}) - \bar Y(\bm z_{j^\prime}) \}.
$$
The mean and covariance of $\bar{\bm Y}^\textrm{obs}$ are respectively
$$
\mathrm E \left( \bar{\bm Y}^\textrm{obs} \right) = \bar{\bm Y},
\quad
\textrm{Cov} ( \bar{\bm Y}^\textrm{obs} ) = 
\begin{bmatrix}
    \left(\frac{1}{n_1} - \frac{1}{N} \right) S^2(\bm z_1)  & -\frac{1}{N} S(\bm z_1, \bm z_2) & \dots  & - \frac{1}{N} S(\bm z_1, \bm z_{2^K}) \\
   -\frac{1}{N} S(\bm z_2, \bm z_1) & \left(\frac{1}{n_2} - \frac{1}{N} \right) S^2(\bm z_2) & \ldots & -\frac{1}{N} S(\bm z_2, \bm z_{2^K}) \\
    \vdots  & \vdots & \ddots & \ldots \\
    -\frac{1}{N} S(\bm z_{2^K}, \bm z_1)  & \dots  & \ldots & \left(\frac{1}{n_{J}} - \frac{1}{N} \right) S^2(\bm z_{2^K}) \\
\end{bmatrix}.
$$
\end{lemma}

\begin{prop}\label{prop:var}
$\hat {\bm \tau}_\textrm{RI}$ is unbiased, and its sampling covariance is
\begin{equation}\label{eq:var}
\mathrm{Cov} ( \hat{\bm \tau}_\textrm{RI} ) 
= \frac{1}{2^{2(K-1)}} 
\sum_{j=1}^{2^K} \frac{1}{n_j} \tilde{\bm h}_j^\prime \tilde{\bm h}_j S^2(\bm z_j) 
- \frac{1}{N(N-1)} \sum_{i=1}^N (\bm \tau_i - \bm \tau)(\bm \tau_i - \bm \tau)^\prime,
\end{equation}
\end{prop}

\begin{proof}[Proof of Proposition \ref{prop:var}]
The unbiasedness of $\hat {\bm \tau}_\textrm{RI}$ is direct from the first half of Lemma \ref{lemma:cov}. Next we derive the covariance. On the one hand, by the second half of Lemma \ref{lemma:cov},
\begin{eqnarray}\label{eq:var-1}
\mathrm{Cov} ( \hat{\bm \tau}_\textrm{RI} ) 
& = & \frac{1}{2^{2(K-1)}} \bm H^\prime \textrm{Cov} ( \bar{\bm Y}^\textrm{obs} ) \bm H \nonumber \\
& = & \frac{1}{2^{2(K-1)}} 
\left\{ 
\sum_{j=1}^{2^K} \tilde{\bm h}_j^\prime \tilde{\bm h}_j \left(\frac{1}{n_j} - \frac{1}{N} \right) S^2(\bm z_j) - \frac{1}{N} \sum_{j \ne j^\prime} \tilde{\bm h}_j^\prime \tilde{\bm h}_{j^\prime} S(\bm z_j, \bm z_{j^\prime})
\right\}.
\end{eqnarray}
On the other hand, by \eqref{eq:ie} and \eqref{eq:fe} we have
\begin{eqnarray}
\frac{1}{N-1} \sum_{i=1}^N (\bm \tau_i - \bm \tau) (\bm \tau_i - \bm \tau)^\prime
& = & \frac{1}{2^{2(K-1)}} \bm H^\prime \left\{ \frac{1}{N-1} \sum_{i=1}^N (\bm Y_i - \bar{\bm Y}) (\bm Y_i - \bar{\bm Y})^\prime \right\} \bm H \nonumber\\
& = & \frac{1}{2^{2(K-1)}} 
\left\{ 
\sum_{j=1}^{2^K} \tilde{\bm h}_j^\prime \tilde{\bm h}_j S^2(\bm z_j) + \sum_{j \ne j^\prime} \tilde{\bm h}_j^\prime \tilde{\bm h}_{j^\prime} S(\bm z_j, \bm z_{j^\prime})
\right\},
\end{eqnarray}
which implies that
\begin{eqnarray}\label{eq:var-2}
\sum_{j \ne j^\prime} \tilde{\bm h}_j^\prime \tilde{\bm h}_{j^\prime} S(\bm z_j, \bm z_{j^\prime})
& = & \frac{2^{2(K-1)}}{N-1} \sum_{i=1}^N (\bm \tau_i - \bm \tau) (\bm \tau_i - \bm \tau)^\prime - \sum_{j=1}^{2^K} \tilde{\bm h}_j^\prime \tilde{\bm h}_j S^2(\bm z_j).
\end{eqnarray}
Substitute the last term in \eqref{eq:var-1} with the right hand side of \eqref{eq:var-2}, we have \eqref{eq:var}. 
\end{proof}

To estimate \eqref{eq:var}, we substitute 
$
S^2(\bm z)
$ 
by its unbiased estimator \citep{Cochran:1977}:
$$
s^2(\bm z_j) = \frac{1}{n_j-1} \sum_{i:W_i(\bm z_j)=1} \{ Y_i^\textrm{obs} - \bar Y^\textrm{obs}(\bm z_j) \}^2,
$$
ignore the second term in the right hand side of \eqref{eq:var}, and obtain the ``Neymanian'' estimator:
\begin{equation}\label{eq:vest}
\widehat \mathrm{Cov}_\textrm{Ney} ( \hat{\bm \tau}_\textrm{RI} ) = \frac{1}{2^{2(K-1)}} \sum_{j=1}^{2^K} \frac{1}{n_j}\tilde{\bm h}_j^\prime \tilde{\bm h}_j s^2(\bm z_j),
\end{equation}
whose bias is
\begin{eqnarray*}
\mathrm{E} \left\{ \widehat \mathrm{Cov}_\textrm{Ney} ( \hat{\bm \tau}_{\textrm{RI}} ) \right\} - \mathrm{Cov} ( \hat{\bm \tau}_{\textrm{RI}} ) =  \frac{1}{N(N-1)} \sum_{i=1}^N (\bm \tau_i - \bm \tau)(\bm \tau_i - \bm \tau)^\prime.
\end{eqnarray*}
In particular, the variance estimator of each component of $\hat {\bm \tau}_\textrm{RI}$ is ``conservative'' \citep{Imbens:2015}, because it always has a nonnegative bias.

\section{The Equivalence between Randomization-based and Regression-based Inferences}\label{sec:reg}

\subsection{Regression-based Inference}

Unlike randomization-based inference, the regression-based inference framework treats the the observed outcome as the ``dependent variable'' of a linear model and the treatment factors (along with their interactions) as ``independent variables.'' To formally define the regression-based estimator, without loss of generality, we assign the first $n_1$ units to treatment $\bm z_1,$ the next $n_2$ units to $\bm z_2$ et. al., which implies the following ``regression'' matrix:
$$
\bm X = 
(\tilde{\bm x}_1^\prime, \ldots, \tilde{\bm x}_{2^K}^\prime)^\prime
=
(
\underbrace{\tilde{\bm h}_1^\prime, \ldots, \tilde{\bm h}_1^\prime}_{\mathrm{repeat} \; n_1 \; \mathrm{times}}
,
\ldots
,
\underbrace{\tilde{\bm h}_{2^K}^\prime, \ldots, \tilde{\bm h}_{2^K}^\prime}_{\mathrm{repeat} \; n_{2^K} \; \mathrm{times}}
)^\prime.
$$
We define the regression-based estimator as:
\begin{equation}
\label{eq:reg}
\hat {\bm \tau}_\textrm{OLS} = 2 \hat {\bm \beta}_\textrm{OLS} = 2 (\bm X^\prime \bm X)^{-1} \bm X^\prime \bm Y^\textrm{obs}.
\end{equation}
To quantify the uncertainty of $\hat {\bm \tau}_\textrm{OLS},$ we consider the following amended Huber-White covariance estimator \citep{MacKinnon:1985}:
\begin{equation}\label{eq:vestreg}
\widehat \mathrm{Cov}_{\textrm{HW}} ( \hat {\bm \tau}_\textrm{OLS} )
= 4N (\bm X^\prime \bm X)^{-1} \left( \frac{1}{N} \sum_{i=1}^N \tilde{\bm x}_i^\prime \tilde{\bm x}_i \frac{\hat e_i^2}{1-\tilde{\bm x}_i (\bm X^\prime \bm X)^{-1} \tilde{\bm x}_i^\prime}\right) (\bm X^\prime \bm X)^{-1},
\end{equation}
where $\hat e_i$ is the estimated residual of unit $i.$ 

\subsection{The Equivalence}

To demonstrate the equivalence between randomization-based and regression-based inferences, we first show the point-wise equivalence between the randomization-based and regression-based estimators. Although this is a well-known result, we provide a direct proof for completeness.

\begin{prop}
\label{prop:pequi}
The randomization-based and regression-based estimators of $\bm \tau$ are point-wisely equivalent, i.e., 
$
\hat {\bm \tau}_\textrm{RI}  = \hat {\bm \tau}_\textrm{OLS}. 
$
\end{prop}

\begin{proof}[Proof of Proposition \ref{prop:pequi}]
On the one hand, 
$$
\bm X^\prime \bm Y^\textrm{obs} = \sum_{i=1}^N \tilde{\bm x}_i^\prime Y_i^\textrm{obs} = \sum_{j=1}^{2^K} n_j \tilde{\bm h}_j^\prime \bar Y^\textrm{obs}(\bm z_j).
$$
On the other hand,
$$
\bm H^\prime \bar{\bm Y}^\textrm{obs} = \sum_{j=1}^{2^K} \tilde{\bm h}_j^\prime \bar Y^\textrm{obs}(\bm z_j).
$$
Moreover, by \eqref{eq:orth},
\begin{eqnarray*}
\bm X^\prime \bm X \bm H^\prime \bar{\bm Y}^\textrm{obs} 
& = & \left( \sum_{j=1}^{2^K} n_j \tilde{\bm h}_j^\prime \tilde{\bm h}_j \right) \left( \sum_{j=1}^{2^K} \tilde{\bm h}_j^\prime \bar Y^\textrm{obs}(\bm z_j)\right) \\
& = & 2^K \left( \sum_{j=1}^{2^K} n_j \tilde{\bm h}_j^\prime \bar Y^\textrm{obs}(\bm z_j) \right) \\
& = & 2^K \bm X^\prime \bm Y^\textrm{obs}.
\end{eqnarray*}
Therefore
\begin{equation*}
\frac{1}{2^K} \bm H^\prime \bar{\bm Y}^\textrm{obs} = (\bm X^\prime \bm X )^{-1} \bm X^\prime \bm Y^\textrm{obs},
\end{equation*}
which completes the proof.
\end{proof}

To show the equivalence between the randomization-based and regression-based confidence regions, we rely on the following lemmas.

\begin{lemma}\label{lemma:matrix}
The regression matrix has the following properties:
$$
(\bm X^\prime \bm X)^{-1} \tilde{\bm h}_j^\prime \tilde{\bm h}_j (\bm X^\prime \bm X)^{-1} = \frac{1}{2^{2K} n_j^2} \tilde{\bm h}_j^\prime \tilde{\bm h}_j.
$$
\end{lemma}

\begin{proof}[Proof of Lemma \ref{lemma:matrix}]
By \eqref{eq:orth},
\begin{eqnarray*}
\bm X^\prime \bm X \tilde{\bm h}_j^\prime \tilde{\bm h}_j 
& = & n_j \tilde{\bm h}_j^\prime \tilde{\bm h}_j \tilde{\bm h}_j^\prime \tilde{\bm h}_j + \sum_{j^\prime \ne j} n_{j^\prime} \tilde{\bm h}_{j^\prime}^\prime \tilde{\bm h}_{j^\prime} \tilde{\bm h}_j^\prime \tilde{\bm h}_j \\
& = &  2^K n_j \tilde{\bm h}_j^\prime \tilde{\bm h}_j.
\end{eqnarray*}
Similarly,
\begin{eqnarray*}
\bm X^\prime \bm X \tilde{\bm h}_j^\prime \tilde{\bm h}_j \bm X^\prime \bm X 
& = &  2^K n_j \tilde{\bm h}_j^\prime \tilde{\bm h}_j \bm X^\prime \bm X \\
& = & 2^{2K} n_j^2  \tilde{\bm h}_j^\prime \tilde{\bm h}_j,
\end{eqnarray*}
which completes the proof.
\end{proof}

\begin{lemma}\label{lemma:leverage}
If unit $i$ is assigned treatment $\bm z_j,$ its ``leverage'' is:
$$
\tilde{\bm x}_i (\bm X^\prime \bm X)^{-1} \tilde{\bm x}_i^\prime = \tilde{\bm h}_j (\bm X^\prime \bm X)^{-1} \tilde{\bm h}_j^\prime = \frac{1}{n_j}.
$$

\end{lemma}

\begin{proof}[Proof of Lemma \ref{lemma:leverage}]
On the one hand, by Lemma \ref{lemma:matrix} and \eqref{eq:orth},
\begin{eqnarray*} 
\{ \tilde{\bm h}_j (\bm X^\prime \bm X)^{-1} \tilde{\bm h}_j^\prime \}^2 
& = & \tilde{\bm h}_j \left\{ (\bm X^\prime \bm X)^{-1} \tilde{\bm h}_j^\prime \tilde{\bm h}_j (\bm X^\prime \bm X)^{-1} \right\}\tilde{\bm h}_j^\prime \\
& = & \frac{1}{2^{2K} n_j^2} \tilde{\bm h}_j \tilde{\bm h}_j^\prime \tilde{\bm h}_j \tilde{\bm h}_j^\prime \\
& = & \frac{1}{n_j^2}.
\end{eqnarray*}
On the other hand, $\tilde{\bm h}_j (\bm X^\prime \bm X)^{-1} \tilde{\bm h}_j^\prime \ge 0.$ The proof is complete.
\end{proof}

\begin{prop}
\label{prop:vequi}
The covariance estimators of the randomization-based and regression-based estimators of $\bm \tau$ are equivalent, i.e., 
$
\widehat \mathrm{Cov}_{\textrm{Ney}} ( \hat {\bm \tau}_\textrm{RI} ) = \widehat \mathrm{Cov}_{\textrm{HW}} ( \hat {\bm \tau}_\textrm{OLS} ).
$
\end{prop}

\begin{proof}[Proof of Proposition \ref{prop:vequi}]
If unit $i$ is assigned treatment $\bm z_j,$ its estimated residual is
\begin{eqnarray*}
\hat e_i 
& = & Y_i^\textrm{obs} - \tilde{\bm h}_j (\bm X^\prime \bm X)^{-1} \bm X^\prime \bm Y^\textrm{obs} \\
& = & Y_i^\textrm{obs} - 2^{-K} \tilde{\bm h}_j \bm H^\prime \bar{\bm Y}^\textrm{obs} \\
& = & Y_i^\textrm{obs} - \bar Y^\textrm{obs} (\bm z_j).
\end{eqnarray*}
The second equation holds by Proposition \ref{prop:pequi}. Therefore, by grouping the units by their treatment assignments and apply Lemmas \ref{lemma:matrix} and \ref{lemma:leverage}, we have
\begin{eqnarray*}
\widehat \mathrm{Cov}_{\textrm{HW}} ( \hat {\bm \tau}_\textrm{OLS} )
& = & 
4 (\bm X^\prime \bm X)^{-1} 
\left\{
\sum_{j=1}^{2^K} \sum_{i: W_i(\bm z_j) = 1} \tilde{\bm h}_j^\prime \tilde{\bm h}_j \frac{\hat e_i^2}{1-\tilde{\bm h}_j 
(\bm X^\prime \bm X)^{-1} \tilde{\bm h}_j^\prime} 
\right\}
(\bm X^\prime \bm X)^{-1} \\
& = & 
4 \sum_{j=1}^{2^K} \frac{(\bm X^\prime \bm X)^{-1} \tilde{\bm h}_j^\prime \tilde{\bm h}_j (\bm X^\prime \bm X)^{-1}}{1-\tilde{\bm h}_j(\bm X^\prime \bm X)^{-1} \tilde{\bm h}_j^\prime} \sum_{i: W_i(\bm z_j) = 1} \{ Y_i^\textrm{obs} - \bar Y^\textrm{obs}(\bm z_j) \}^2 \\
& = & 
\frac{4}{2^{2K}}
\sum_{j=1}^{2^K} \frac{\tilde{\bm h}_j^\prime \tilde{\bm h}_j}{n_j(n_j - 1)}\sum_{i: W_i(\bm z_j) = 1} \{ Y_i^\textrm{obs} - \bar Y^\textrm{obs}(\bm z_j) \}^2 \\
& = & \frac{1}{2^{2(K-1)}} \sum_{j = 1}^{2^K} \frac{1}{n_j}\tilde{\bm h}_j^\prime \tilde{\bm h}_j s^2(\bm z_j) \\
& = & \widehat \mathrm{Cov}_{\textrm{Ney}} ( \hat {\bm \tau}_\textrm{RI} ).
\end{eqnarray*}
\end{proof}

\section{Extensions}\label{sec:extension}

For balanced designs, it is usually possible to simplify the derivations of the covariance estimators \citep{Samii:2012}. Consequently, researchers often assume balanced designs \citep[e.g.,][]{Gadbury:2001, Dasgupta:2015}. In balanced $2^K$ factorial designs, we allow $N=2^K r$ $(r \ge 2)$ experiment units and assign $r$ to each treatment; Propositions \ref{prop:var} and \ref{prop:vequi} then reduce to the main results of \cite{Dasgupta:2015}, summarized in the following corollary. 

\begin{corollary}
In balanced $2^K$ factorial designs with $N = 2^K r$ experimental units,
\begin{equation}\label{eq:var-balanced}
\widehat \mathrm{Cov}_\textrm{Ney} ( \hat{\bm \tau}_\textrm{RI} ) 
= \widehat \mathrm{Cov}_{\textrm{HW}} ( \hat {\bm \tau}_\textrm{OLS} )
= \frac{1}{2^{2(K-1)}r} \sum_{j=1}^{2^K} \tilde{\bm h}_j^\prime \tilde{\bm h}_j s^2(\bm z_j). 
\end{equation}
\end{corollary}

Next, we discuss covariance estimation in standard regression analysis, which is commonly used by practitioners \citep[e.g.,][]{Chakraborty:2009, Collins:2009}. Under the homoscedasticity assumption, we let
\begin{equation}\label{eq:linear-model}
\bm Y^\textrm{obs} = \frac{1}{2} \bm X \bm \tau + \bm \epsilon,
\quad
\bm \epsilon \sim N(\bm 0, \sigma^2 \bm I_N),
\end{equation}
which leads to the following ``standard'' covariance estimator
\begin{equation}\label{eq:var-balanced-2}
\widehat \mathrm{Cov}_\textrm{HE} ( \hat {\bm \tau}_\textrm{OLS} ) 
= 4 \hat \sigma^2 (\bm X^\prime \bm X)^{-1}
= \frac{4}{N-2^K} \sum_{i=1}^N \hat e_i^2 (\bm X^\prime \bm X)^{-1}.
\end{equation}
Straightforward arithmetic suggests that \eqref{eq:vestreg} and \eqref{eq:var-balanced-2} are different, even in balanced $2^K$ factorial designs. However, the following corollary suggests that in balanced $2^K$ factorial designs \eqref{eq:var-balanced-2} does give correct variance estimator of any component of $\hat{\bm \tau}_\textrm{OLS}.$

\begin{corollary}\label{coro:var}
In balanced $2^K$ factorial designs with $N = 2^K r$ experimental units, the variance estimator of any component of $\hat{\bm \tau}_\textrm{OLS}$ under \eqref{eq:linear-model} is equivalent to that in \eqref{eq:var-balanced}. 
\end{corollary}

\begin{proof}[Proof of Corollary \ref{coro:var}] 
For balanced $2^K$ factorial designs, $\bm X^\prime \bm X = N \bm I_{2^K}.$ Therefore
\begin{eqnarray}\label{eq:coro-var}
\widehat \mathrm{Cov}_\textrm{HE} ( \hat {\bm \tau}_\textrm{OLS} ) 
& = & (\bm X^\prime \bm X)^{-1} \frac{1}{2^{K-2}} \sum_{j=1}^{2^K} \frac{1}{r-1} \sum_{i: W_i(\bm z_j) = 1} \{ Y_i^\textrm{obs} - \bar Y^\textrm{obs}(\bm z_j) \}^2 \nonumber \\
& = & \frac{\bm I_{2^K}}{2^{2(K-1)}r} \sum_{j=1}^{2^K} s^2(\bm z_j).
\end{eqnarray}

Let $\hat {\bm \tau}_\textrm{OLS} = \{ \hat{\bm \tau}_\textrm{OLS} (1), \ldots, \hat{\bm \tau}_\textrm{OLS} (2^K) \}^\prime.$ Therefore, the variance estimator  $\hat{\bm \tau}_\textrm{OLS} (j)$ under \eqref{eq:linear-model} is the $j$th diagonal element of the covariance estimator matrix in \eqref{eq:coro-var}:
$$
\widehat \mathrm{Var}_\textrm{HE} \{ \hat {\bm \tau}_\textrm{OLS} (j) \} = \frac{1}{2^{2(K-1)}r} \sum_{j=1}^{2^K} s^2(\bm z_j),
$$
which equals the $j$th diagonal element of the covariance estimator matrix in  \eqref{eq:var-balanced}. 
\end{proof}

From a practical perspective, in applied analyses of balanced $2^K$ factorial designs where we are only interested in estimating the variances of factorial effects, Corollary \ref{coro:var} assures the validities of standard regression methods.

\section{Concluding Remarks}\label{sec:discuss}

In this paper, we demonstrate the equivalence between randomization-based and regression-based inferences for $2^K$ factorial designs. As pointed out by \cite{Samii:2012}, while regression-based methods may not be favorable for randomized experiments, they in fact can be justified by randomization. Our results show that practitioners can use regressions for $2^K$ factorial designs, as long as they also use the amended Huber-White estimator to quantify uncertainties of the estimated factorial effects.

Our work implies multiple future directions. First, we can generalize our current framework to other factorial designs such as $3^k$ factorial designs or fractional factorial designs. Second, it is possible to further unify randomization-based and regression-based inferences with Bayesian inference. Third, we can incorporate pretreatment covariates to our analysis.

\section*{Acknowledgements}

The author thanks Professors Tirthankar Dasgupta and Joseph K. Blitzstein at Harvard, for their invaluable mentorships, and Professor Peng Ding at UC Berkeley, for helpful suggestions. Thoughtful comments from the Co-Editor-in-Chief, an Associate Editor and a reviewer substantially improved the quality of the paper.

\bibliographystyle{apalike}
\bibliography{factorial_hw}

\end{document}